\documentclass{article}
\usepackage[utf8]{inputenc} 
\usepackage[T1]{fontenc}    
\usepackage{hyperref}       
\usepackage{url}
\usepackage{booktabs}       
\usepackage{amsfonts}       
\usepackage{nicefrac}       
\usepackage{microtype}      
\usepackage{amsmath,amsthm,amssymb} 
\usepackage{enumerate}       
\usepackage{color}
\usepackage{algorithm}
\usepackage{graphicx}
\usepackage{multirow}

\newtheorem{theorem}{Theorem}
\newtheorem{lemma}{Lemma}
\DeclareMathOperator*{\argmax}{argmax}
\DeclareMathOperator*{\argmin}{argmin}
\DeclareMathOperator*{\minimize}{minimize}
\DeclareMathOperator*{\Cov}{Cov}

\usepackage[accepted]{shapley_arxiv}

\shapleytitlerunning{Efficient inference on population feature importance using Shapley values}

\begin{document}

\twocolumn[
\shapleytitle{Efficient nonparametric statistical inference on population feature importance \\ using Shapley values}

\shapleysetsymbol{equal}{*}

\begin{shapleyauthorlist}
\shapleyauthor{Brian D.~Williamson}{equal,1}
\shapleyauthor{Jean Feng}{equal,2}
\end{shapleyauthorlist}

\shapleyaffiliation{1}{Vaccine and Infectious Disease Division, Fred Hutchinson Cancer Research Center, Seattle, WA}
\shapleyaffiliation{2}{Department of Biostatistics, University of Washington, Seattle, WA}

\shapleycorrespondingauthor{Brian D.~Williamson}{bwillia2@fredhutch.org}

\shapleykeywords{Variable importance, Shapley values, nonparametric statistics, statistical inference}

\vskip 0.3in
]

\printAffiliationsAndNotice{\shapleyEqualContribution}

\begin{abstract}
    The true population-level importance of a variable in a prediction task provides useful knowledge about the underlying data-generating mechanism and can help in deciding which measurements to collect in subsequent experiments.
    Valid statistical inference on this importance is a key component in understanding the population of interest.
    We present a computationally efficient procedure for estimating and obtaining valid statistical inference on the \textbf{S}hapley \textbf{P}opulation \textbf{V}ariable \textbf{I}mportance \textbf{M}easure (SPVIM).
    Although the computational complexity of the true SPVIM scales exponentially with the number of variables, we propose an estimator based on randomly sampling only $\Theta(n)$ feature subsets given $n$ observations.
    We prove that our estimator converges at an asymptotically optimal rate.
    Moreover, by deriving the asymptotic distribution of our estimator, we construct valid confidence intervals and hypothesis tests.
    Our procedure has good finite-sample performance in simulations, and for an in-hospital mortality prediction task produces similar variable importance estimates when different machine learning algorithms are applied.
\end{abstract}

\section{Introduction}\label{sec:intro}

In many scientific applications, understanding the intrinsic predictive value of a variable can shed light on the internal mechanisms relating the variable to the outcome of interest, help build future models, and guide experimental design.
For example, hospital administrators may want to know the important features to collect for predicting patient outcomes.
Likewise, vaccine researchers may want to know the most important molecular phenotypes to measure that are most predictive of binding or vaccine efficacy \citep[see, e.g.,][]{dunning2006}.
Variable importance measures (VIMs) provide necessary information towards answering these questions.

Our interest here is in statistical inference on the population VIM.
This VIM quantifies the predictive value of a variable within the oracle prediction model $f_{0}$ defined relative to an arbitrary predictiveness measure $V$.
For many choices of $V$, $f_0$ is either the conditional mean outcome given covariates (e.g., if $V = R^2$) or a simple functional of this conditional mean (e.g., if $V = $ classification accuracy).
We note that population VIMs are distinct from algorithmic VIMs, which describe the importance of a variable within a fitted model $\hat{f}$ \citep[see, e.g.,][]{breiman2001, garson1991, murdoch2019}.
Although algorithmic VIMs have been used as a proxy for population VIMs out of convenience, differences between $\hat{f}$ and $f_0$ can often lead to substantially different interpretations of the resulting VIMs.
Whereas an algorithmic VIM necessarily varies across fitted models, a population VIM is independent of the specific procedure used to estimate $f_0$.

Existing population VIMs suffer from a number of issues.
Traditionally, population VIMs have relied on restrictive parametric assumptions \citep[e.g., $R^2$ in linear models; see, e.g.,][]{gromping2007,nathans2012}, which can lead to misleading results if the parametric model does not hold.
Recent work has focused on extending these definitions by removing the parametric assumptions \citep{williamson2020b}; however, these definitions define importance of a variable with respect to the others and assign near-zero importance when features are highly correlated.
Other VIMs require strong assumptions on the design to be valid (e.g., ANOVA), but again fail in simple cases with correlated variables.
To address this, \citet{owen2017} proposed using Shapley values to quantify the population VIM, where the value function is the variance explained; these VIMs inherit many desirable theoretical properties from the Shapley value.
In fact, contemporary work has also defined the ideal estimand of algorithmic VIM estimation procedures to be the Shapley population VIM (SPVIM) \citep{covert2020}.

Unfortunately, exact estimation of SPVIM is computationally intractable in general settings \citep{owen2017}: the SPVIM is defined as the sum of $2^p$ terms, where $p$ is the number of features and each term depends on estimating the conditional mean function with respect to a unique feature subset.
Previous approaches have either suggested sampling as many subsets as possible to estimate the Shapley value \citep[see, e.g.,][]{castro2009} or utilized special properties of tree estimators to reduce the number of subsets required \citep{lundberg2020}.
Notably, \citet{strumbelj2014} analyzed the asymptotic distribution of a sampling-based estimator of Shapley algorithmic variable importance to derive confidence intervals.

In this paper, we combine the aforementioned developments and provide a nonparametric statistical inference procedure for SPVIM.
We generalize previous definitions of SPVIM and use an arbitrary measure of predictiveness $V$.
We tackle the computational complexity of the problem by randomly sampling feature subsets according to the Shapley value weights and then fitting corresponding models.
We derive the asymptotic distribution of this sampling-based SPVIM estimator and show that the error from our proposed procedure can be decomposed into two components: the error from estimating the oracle prediction models and the error from omitting summands from the Shapley value estimand.
Given $n$ training observations, we find that our estimator only needs to sample $m = \Theta(n)$ subsets to converge at an asymptotically optimal rate.
Moreover, since the subset sampling distribution is highly skewed, the number of \textit{unique} feature subsets is much smaller than $m$ in practice.
We then use the asymptotic distribution to construct asymptotically unbiased point estimates, valid confidence intervals, and hypothesis tests with proper type I error control.

We demonstrate the validity of our approach in a simulation study and estimate the SPVIM of hospital measurements for predicting mortality in the intensive care unit (ICU).
All numerical results can be replicated using code available on GitHub at \texttt{bdwilliamson/spvim\_supplementary}; the proposed methods are also implemented in the Python package \texttt{vimpy} and the R package \texttt{vimp}.

\section{Variable importance}
\label{sec:vimp}
\subsection{Data structure and notation}\label{sec:prelim}
Let $\mathcal{M}$ be a nonparametric class of joint distributions over covariates $X = (X_1, \ldots, X_p) \in \mathcal{X}  \subseteq \mathbb{R}^p$ and response $Y \in \mathcal{Y} \subseteq \mathbb{R}$, where $\mathcal{X}$ and $\mathcal{Y}$ denote the sample spaces of $X$ and $Y$, respectively.
Suppose that each observation $O$ consists of $(X, Y)$.
In this article, we consider observations $O_1,\dots,O_n$ drawn independently according to a joint probability distribution $P_0 \in \mathcal{M}$.

Next, we define the feature subsets and oracle prediction models of interest.
We take $\mathcal{S}$ to be the power set of $N := \{1, \ldots, p\}$.
Let $s_{(j)}$ for $j = 1,\ldots,2^p$ be an ordered sequence of the subsets in $\mathcal{S}$, where $s_{(1)} = \emptyset$ and $s_{(2^p)} = N$.
For any index set $s \in \mathcal{S}$, we denote by $\mathcal{X}_s$ and $\mathcal{X}_{-s}$ the sample spaces of $X_s$ and $X_{-s}$, respectively.
We denote by $a_s$ and $a_{-s}$ the elements of a vector $a$ with indices in $s$ and not in $s$, respectively.
We also consider the binary vector $z(s) \in \mathbb{R}^{p+1}$ for each $s \in \mathcal{S}$, where $z(s)_1 = 1$ for all $s \in \mathcal{S}$ and $z(s)_{k+1} = I(k \in s)$ for $k = 1, \ldots, p$.
Finally, we consider a rich class $\mathcal{F}$ of functions from $\mathcal{X}$ to $\mathcal{Y}$ endowed with a norm $\lVert \cdot \rVert_\mathcal{F}$.
For any $s \in \mathcal{S}$, we define the subset $\mathcal{F}_s := \{f \in \mathcal{F}: f(u) = f(v) \text{ for all } u, v \in \mathcal{X} \text{ satisfying } u_s = v_s\}$ of functions in $\mathcal{F}$ whose evaluation ignores elements of the input $x$ with index not in $s$.
In all examples we consider, we take $\mathcal{F}$ to be a rich class of functions that is essentially unrestricted up to regularity conditions.

\subsection{Oracle predictiveness}\label{sec:pred}
We define the importance of a variable at the population level in terms of its oracle predictiveness.
This predictiveness is measured by a real-valued functional $V: \mathcal{F} \times \mathcal{M} \mapsto \mathbb{R}$.
We assume that larger values of $V(f, P)$ imply higher predictiveness.
Examples of predictiveness measures --- including $R^2$, deviance, area under the ROC curve, and classification accuracy --- are provided in \citet{williamson2020b}.

The oracle predictiveness is the maximum achievable predictiveness over a class of prediction functions.
More formally, we define the \textit{total oracle predictiveness}
$
v_{0,N} := \max_{f \in \mathcal{F}} V(f, P_0)
$
and its associated oracle prediction function $f_{0,N} \in \argmax_{f \in \mathcal{F}} V(f, P_0)$.
For many machine learning algorithms, $f_{0,N}$ is the target of interest.
We further define the oracle prediction function $f_{0,s}$ that maximizes $V(f, P_0)$ over all $f \in \mathcal{F}_s$; the \textit{marginal oracle predictiveness} $v_{0, s} := V(f_{0,s}, P_0)$ quantifies the prediction potential of features with index in $s$.
The \textit{null oracle predictiveness} $v_{0,\emptyset} := V(f_{0,\emptyset}, P_0)$ quantifies the prediction potential of a model that uses no covariate information.
Finally, let $v_0 := [v_{0,\emptyset}, v_{0,s_{(2)}},\ldots,v_{0,N}]^\top$ denote the $2^p$-dimensional vector of predictiveness measures for all subsets in $\mathcal{S}$.
The predictiveness measure $v_{0,s_{(j)}}$ is defined relative to the population $P_0$, a joint distribution in the nonparametric statistical model $\mathcal{M}$; thus, its interpretation is tied to neither any particular estimation procedure nor any parametric assumptions.

\subsection{The Shapley population variable importance measure}\label{sec:shap}

We now define a population VIM using the classical form of the Shapley value \citep[see, e.g.,][]{shapley1953,charnes1988} with an arbitrary measure of predictiveness $V$.
Specifically, the \textit{Shapley population variable importance measure} (SPVIM) of the variable $X_j$ is the average gain in oracle predictiveness from including feature $X_j$ over all possible subsets:
\begin{align}\label{eq:class_shap}
\hspace{-0.1in}
    \psi_{0, 0,j} &:= \sum_{s \in N \setminus \{j\}}\frac{1}{p} \binom{p-1}{\lvert s \rvert}^{-1}\{V(f_{0, s\cup j}, P_0) - V(f_{0,s}, P_0)\},
\end{align}
where the indices of $\psi$ describe the number of subsets, the distribution $P_0$, and the feature of interest $j$, respectively.
We use the index 0 to indicate that the SPVIM is computed using all subsets and the true distribution $P_0$.
SPVIMs inherit the following properties from Shapley values \citep{shapley1953}:
\begin{itemize}
	\item Non-negativity: by construction, $\psi_{0,0,j} \geq 0$.
	\item Additivity\footnote{In the Shapley value literature, this additivity property is referred to as ``efficiency''. However, this notion of efficiency is very different from statistical efficiency, which is related to the asymptotic variance of a statistical estimator.}: the sum of the SPVIMs across all variables is equal to the difference between the total and null oracle predictiveness,
		\begin{align}\label{eq:additivity}
	       \sum_{j=1}^p \psi_{0,0,j} = v_{0,N} - v_{0, \emptyset}
    \end{align}
	\item Symmetry: if $X_i = X_j$, then $\psi_{0,0,i} = \psi_{0,0,j}$.
	\item Null feature: if $X_j$ provides no added predictive value, i.e., $v_{0,s\cup j} = v_{0,s}$ for all $s \subseteq (N \setminus \{j\})$, then its SPVIM value is $\psi_{0,0,j} = 0$.
	\item Linearity: if $\tilde{V} \equiv \alpha V$, then its associated SPVIM values are $\tilde{\psi}_{0,0,j} = \alpha \psi_{0,0,j}$ for all $j \in \{1, \ldots, p\}$.
\end{itemize}
Because SPVIMs satisfy these properties, they clearly address the issue of correlated features: given collinear variables $X_{j}$ and $X_{k}$ that are each marginally predictive, previous nonparametric population VIMs \citep[see, e.g., ][]{williamson2020b} would assign zero importance to both variables whereas SPVIM would assign the same positive value to both variables.

In this paper, we take advantage of an alternate formulation of the Shapley value noted in previous work \citep[see, e.g.,][]{charnes1988, lundberg2017}.
In particular, we can rewrite the weighted average in \eqref{eq:class_shap} as the solution of a weighted linear regression problem, where we treat the predictiveness of a feature subset $v_{0,s}$ as the response and the subset membership $z(s)$ as the covariates.
Define a diagonal matrix of weights $W \in \mathbb{R}^{2^p \times 2^p}$ where $W_{1,1} = W_{2^p,2^p} = 1$, and for any $j \in 2, \ldots, 2^p - 1$, $W_{j,j} = \binom{p-2}{|s_{(j)}|-1}^{-1}$.
The matrix $Z \in \mathbb{R}^{2^p \times (p+1)}$ consists of the stacked $z(s)$ vectors for each $s \in \mathcal{S}$.
Setting $\psi_{0, 0,\emptyset} := v_{0,\emptyset}$, we denote by $\psi_{0,0}$ the $(p+1)$-dimensional vector of population Shapley values.
Then \eqref{eq:class_shap} is equivalent to
\begin{align}\label{eq:wls}
    \psi_{0, 0} := & \argmin_{\psi \in \mathbb{R}^{p+1}} \lVert \sqrt{W}(Z\psi - v_0) \rVert_2^2,
\end{align}
a result that we prove in the Supplement.
If we define the distribution $Q_0$ over subsets $\mathcal{S}$ with probability mass function assigning weight $\binom{p-2}{\lvert S \rvert - 1}^{-1}$ for $S \in \mathcal{S} \setminus \{\emptyset, N\}$ and weight 1 for $S \in \{\emptyset, N\}$ (scaled so that the weights sum to one), then \eqref{eq:wls} is equivalent to a population average:
\begin{align*}
    \psi_{0,0} \equiv & \argmin_{\psi \in \mathbb{R}^{p+1}} E_{Q_0} \left[
    (z(S)\psi - v_{0,S})^2
    \right] .
\end{align*}
We will use this fact in our estimation procedure below.

\section{Estimation and inference}
\subsection{Plug-in estimation}\label{subsec:plug-in}
We now discuss how to estimate the SPVIM values for all $p$ features using independent observations $O_1, \ldots, O_n$ drawn from $P_0$.
Definition \eqref{eq:wls} suggests considering an estimator based on plugging in estimators of each individual component.
We discuss each component in turn.

First, we estimate the predictiveness measure $v_{0,s} = V(f_{0,s}, P_0)$ for a subset $s \in \mathcal{S}$ by plugging in estimates of the oracle function $f_{0,s}$ and the distribution $P_0$.
A simple approach is to partition the data into a training set and a validation set, construct an estimator $f_{n,s}$ for $f_{0,s}$ on the training data (using only the observed covariates in $s$), and estimate $P_0$ using the empirical distribution of the validation set $P_V$.
Using this training-validation split, our estimate of predictiveness is then
\begin{align}\label{eq:tain_val_pred}
v_{n,s} = V(f_{n,s}, P_V).
\end{align}
An alternative approach is to perform $K$-fold cross-fitting, where we partition the data into $K$ subsets of roughly equal size and, for each $k \in \{1,\ldots,K\}$, construct an estimator $f_{k,n,s}$ based on all the data except for the $k$th subset.
Let $P_{k,n}$ be the empirical distribution of the $k$th subset.
Then we could estimate $v_{0,s}$ using
\begin{align}\label{eq:kfold_pred}
	v_{n,s} = \frac{1}{K}\sum_{k=1}^K V(f_{k,n,s}, P_{k,n}).
\end{align}

If we had the entire estimated vector of predictiveness measures $v_n$, we could estimate $\psi_{0,0}$ using the plug-in estimator
\begin{align}
    \psi_{0,n}  &:= \argmin_{\psi \in \mathbb{R}^{p+1}} E_{Q_0} \left[
    (Z(S) \psi - v_{n,S})^2
    \right].
\label{eq:pop_plugin}
\end{align}
Unfortunately, obtaining $v_n$ requires training $2^p$ models, rendering this a computationally intractable task in general.

Instead, we can replace $Q_0$ in \eqref{eq:pop_plugin} with an empirical distribution estimator $Q_m$ obtained by sampling $m$ subsets from $\mathcal{S}$ according to $Q_0$.
This leads us to the SPVIM estimator $\psi_{m,n}$ which solves the constrained least squares problem
\begin{align}
\min_{\psi \in \mathbb{R}^{p+1}} & E_{Q_m} \left[
(Z(S) \psi - v_{n,S})^2
\right]
\text{subject to } G\psi = c_n,
\label{eq:plugin_shap}
\end{align}
where $G := [z(\emptyset)^\top, z(N)^\top]^\top \in \mathbb{R}^{2 \times (p + 1)}$ and $c_n := [v_{n,\emptyset}, v_{n,N}]^\top \in \mathbb{R}^2$.
The constraint ensures that the estimated SPVIMs satisfy the additivity property~\eqref{eq:additivity} and that the estimated SPVIM for the null set is the estimated null predictiveness value.

This constrained least squares problem can be solved by forming a Lagrangian and inverting its Karush-Kuhn-Tucker (KKT) conditions \citep{boyd2004}.
More specifically, let $s_1,\ldots, s_{\ell}$ be the unique subsets in $Q_m$.
Let $W_m$ be the $\ell\times \ell$ diagonal matrix where the $k$th diagonal element is the probability mass of $s_{k}$ in $Q_m$.
Let $v_{m,n} = (v_{n,s_1}, \ldots, v_{n,s_{\ell}})$ be the estimated predictiveness measures for the $\ell$ subsets.
Let $Z_m$ be the stack of vectors $z(s_1), \ldots, z(s_{\ell})$.
Then \eqref{eq:plugin_shap} can also be written as
\begin{align*}
\min_{\psi \in \mathbb{R}^{p+1}} \left \| \sqrt{W_m} \left( Z_m \psi - v_{m,n} \right ) \right \|_2^2 \text{subject to } G\psi = c_n.
\end{align*}
Solving the KKT conditions with Lagrange multipliers denoted by $\lambda$, we obtain a closed-form SPVIM estimator:
\begin{align}
\left[
\begin{matrix}
\psi_{m,n}\\
\lambda
\end{matrix}
\right]
=
\left[
\begin{matrix}
2 Z_m^\top W_m Z_m & G^\top\\
G & 0
\end{matrix}
\right]
^{-1}
\left[
\begin{matrix}
2 \sqrt{W_m} v_{m,n}\\
c_n
\end{matrix}
\right].
\label{eq:kkt_main}
\end{align}
To ensure that \eqref{eq:plugin_shap} has a unique solution, we select a sufficiently large value of $m$ so that $Q_m$ inclues at least $p + 1$ unique subsets.
The full estimation procedure is given in Algorithm~\ref{alg:spvi_cls}.

We now describe the properties listed in Section~\ref{sec:shap} that are satisfied by this sampling-based SPVIM estimator.
It is easy to see that the additivity, symmetry, and linearity properties always hold.
One possible concern is that the nonnegativity property can be violated.
Nevertheless, in practice we find that negative SPVIM estimates are close to zero and the 95\% confidence intervals cover zero.
If nonnegativity is truly a concern, one can also add a nonnegative constraint to \eqref{eq:plugin_shap}.
Finally, the null feature property holds with respect to \textit{estimated} predictiveness values and the sampled subsets.
Note that this property is only relevant for discrete predictiveness measures like 0-1 classification accuracy, since the estimated predictiveness values are rarely exactly the same for continuous predictiveness measures like $R^2$.

The plug-in estimator $\psi_{m,n}$ is appealing due to its simplicity.
In general, however, such an estimator may fail to be consistent at rate $n^{-1/2}$ if the population optimizers $f_{0,s}$ are flexibly estimated.
This phenomenon is due in large part to the optimal bias-variance tradeoff for estimating $f_{0,s}$ differing in general from the optimal bias-variance tradeoff for estimating $v_{n,s}$.
Plug-in estimators typically inherit much of the bias from estimating $f_{0,s}$, and this bias does not in general tend to zero sufficiently fast to allow $n^{-1/2}$-rate estimation of $\psi_{0,0}$ \citep{williamson2020a}.
In the next section, we extend the results of \citet{williamson2020b} to describe conditions under which the estimator $\psi_{m,n}$ is asymptotically normal.

\begin{algorithm}
    \caption{Estimation of SPVIM}
    \label{alg:spvi_cls}
    \begin{algorithmic}[1]
        \STATE Input initial parameter $\gamma \ge 1$.
        \STATE Sample $m= \gamma n$ subsets from $Q_0$, denoted $s_1,\ldots,s_m$.
        \STATE Estimate prediction functions $f_{n,s}$ for each $s \in \{s_1,\ldots,s_m\} \cup \{\emptyset, N\}$.
        \STATE Compute predictiveness estimates $v_{n,s}$ for $s \in \{s_1,\ldots,s_m\} \cup \{\emptyset, N\}$ using a training-validation split (see Equation \eqref{eq:tain_val_pred}) or $K$-fold cross-fitting (see Equation \eqref{eq:kfold_pred}).
        \STATE Solve for $\psi_{m,n}$ using Equation \eqref{eq:kkt_main}.
    \end{algorithmic}
\end{algorithm}

\subsection{Large-sample inferential properties}\label{sec:inf}

We now study the conditions under which $\psi_{m,n}$ is an asymptotically normal estimator of the SPVIM $\psi_{0,0}$.
Using these conditions, we can design a procedure to construct confidence intervals and hypothesis tests.
To do this, we decompose the error of our estimator $\psi_{m,n}$ into the following components:
\begin{align}\label{eq:decomposition}
    \psi_{m,n} - \psi_{0,0} =& \ (\psi_{0,n} - \psi_{0,0}) + (\psi_{m,0} - \psi_{0,0}) + r_{m,n},
\end{align}
where $\psi_{m,0}$ is obtained by replacing $v_{n,S}$ with $v_{0,S}$ in \eqref{eq:plugin_shap} and $r_{m,n} := (\psi_{m,n} - \psi_{m,0}) - (\psi_{0,n} - \psi_{0,0})$.
Each term on the right-hand side of \eqref{eq:decomposition} can then be studied separately to determine the large-sample behavior of $\psi_{m,n}$.
The first term is the error of the estimator $\psi_{0,n}$ \eqref{eq:pop_plugin} constructed using prediction functions $f_{n,s}$ estimated using $n$ observations for all subsets $s$.
The second term is the error of the estimator $\psi_{m,0}$ constructed using oracle prediction functions for sampled subsets in $Q_m$.
In other words, the first term characterizes the error contribution from sampling training observations and the second term characterizes the error contribution from sampling subsets.
The third term is a difference-in-differences remainder term that we prove to be negligible under some regularity conditions.
Based on this decomposition, we will show that the asymptotic variance of $\sqrt{n}(\psi_{m,n} - \psi_{0,0})$ is simply the sum of the asymptotic variances of the first and second error terms.

Our result makes use of several conditions that require additional notation.
These conditions were initially provided in \citet{williamson2020b}.
We define the linear space $\mathcal{R}:=\{c(P_1-P_2):c\in\mathbb{R},P_1,P_2\in\mathcal{M}\}$ of finite signed measures generated by $\mathcal{M}$.
For any $R\in\mathcal{R}$, e.g., $R=c(P_1-P_2)$, we consider the supremum norm $\|R\|_\infty:=|c|\sup_{o}|F_1(o)-F_2(o)|$, where $F_1$ and $F_2$ are the distribution functions corresponding to $P_1$ and $P_2$, respectively.
Next, we define the following notation for each subset $s \in \mathcal{S}$.
For distribution $P_{0,\epsilon} := P_0 + \epsilon h$ with $\epsilon \in \mathbb{R}$ and $h \in \mathcal{R}$, we define $f_{0,\epsilon,s} = f_{P_{0,\epsilon},s}$ to be its corresponding oracle prediction function with respect to subset $s$.
Let $\dot{V}(f,P_0;h)$ denote the G\^{a}teaux derivative of $P\mapsto V(f,P)$ at $P_0$ in the direction $h\in\mathcal{R}$, and define the random function $g_{n,s}:o\mapsto\dot{V}(f_{n,s},P_0;\delta_o-P_0)-\dot{V}(f_{0,s},P_0;\delta_o-P_0)$, where $\delta_o$ is the degenerate distribution on $\{o\}$.
Consider the following set of deterministic [(A1)--(A4)] and stochastic [(B1)--(B2)] conditions for each subset $s \in \mathcal{S}$:
\begin{enumerate}
     \item[(A1)] (\textit{optimality}) there is some $C>0$ such that for each sequence $f_1,f_2, \dots \in \mathcal{F}_s$ with $\lVert f_j - f_{0,s} \rVert_{\mathcal{F}_s} \to 0$, there is a $J$ such that for all $j > J$, $\lvert V(f_j,P_0) - V(f_{0,s},P_0) \rvert  \leq C\lVert f_j - f_{0,s} \rVert_{\mathcal{F}_s}^2$;
     \item[(A2)] (\textit{differentiability}) there is some $\delta > 0$ such that for each sequence $\epsilon_1,\epsilon_2,\ldots\in\mathbb{R}$ and $h, h_1, h_2, \ldots \in \mathcal{R}$ satisfying that $\epsilon_j\rightarrow 0$ and $\lVert h_j - h \rVert_{\infty} \to 0$, it holds that
     \begin{align*}
         \sup_{f\in\mathcal{F}_s:\lVert f - f_{0,s} \|_{\mathcal{F}_s} < \delta}\bigg\lvert & \frac{V(f, P_0 + \epsilon_j h_j) - V(f, P_0)}{\epsilon_j}  \\
         &- \dot{V}(f,P_0; h_j) \bigg\rvert \longrightarrow 0\ ;
     \end{align*}
     \item[(A3)] (\textit{optimizer continuity}) $\lVert f_{0,\epsilon,s} - f_{0,s} \rVert_{\mathcal{F}_s} = o(\epsilon)$ for each $h \in \mathcal{R}$;
     \item[(A4)] (\textit{derivative continuity}) $f \mapsto \dot{V}(f, P_0; h)$ is continuous at $f_{0,s}$ relative to $\mathcal{F}_s$ for each $h \in \mathcal{R}$;
     \item[(B1)] (\textit{minimum rate of convergence}) $\|f_{n,s}-f_{0,s}\|_{\mathcal{F}_s}=o_P(n^{-1/4})$;
     \item[(B2)] (\textit{weak consistency}) $E_0[\int \{g_{n,s}(o)\}^2dP_0(o)]=o_P(1)$;
\end{enumerate}
The G\^ateaux derivative $\dot{V}$ is provided in \citet{williamson2020b} for several common measures of predictiveness, including classification accuracy, AUC, and $R^2$.
Assuming conditions (A1)--(A4) and (B1)--(B2) hold for every subset in $\mathcal{S}$, $v_n$ is an asymptotically linear estimator of $v_0$ with influence function $\dot{V}_0: o \mapsto [\dot{V}(f_{0,\emptyset}, P_0; \delta_o - P_0),\ldots,\dot{V}(f_{0,N}, P_0; \delta_o - P_0)]^\top$ by Theorem 2 in \citet{williamson2020b}.
Finally, we introduce a condition that specifies the number of subsets to sample:
\begin{enumerate}
    \item[(C1)] (\textit{minimum number of subsets}) For $\gamma > 0$ and sequence $\gamma_1, \gamma_2, \ldots \in \mathbb{R}^+$ satisfying that $\lvert \gamma_j - \gamma \rvert \to 0$, $m = \gamma_n n$.
\end{enumerate}
For convenience, we define several objects that simplify the notation in our main result below.
Set $A := Z^\top W Z$, where $Z$ is the stack of vectors $z(s)$ for all $s \in \mathcal{S}$, and define $C := A^{-1} G (G^\top A^{-1} G )^{-1}$.
Let the QR decomposition of $G^\top$ be
$$
G^\top = \left[\begin{matrix} U_1 & U_2\end{matrix} \right ]
\left[
\begin{matrix}
R \\ 0
\end{matrix}
\right],
$$
where $R$ is an upper-triangular matrix.
We define the functions
\begin{align*}
\phi_{0,1}(O) & = A^{-1} Z^\top \sqrt{W} \dot{V}_0(O) \text{ and }\\
\phi_{0,2}(S; v_0) & = - U_2
V^{-1}
\left[
z(S)^\top
\psi_{0,0}
- v_{0, S}
\right]
U_2^\top z(S),
\end{align*}
where $V = U_2^\top Z^\top W Z U_2 $.
Assuming all of the aforementioned conditions hold, then $\psi_{m,n}$ is a consistent and an asymptotically normal estimator of $\psi_{0,0}$.
\begin{theorem}\label{thm:shapley_vim}
    If the collection of conditions implied by (A1)--(A4) and (B1)--(B2) hold for every subset in $\mathcal{S}$ and condition (C1) holds, then $\psi_{m,n}$ has the asymptotic distribution
    \begin{align*}
    \begin{split}
    & \sqrt{n}(\psi_{m,n} - \psi_{0,0}) \rightarrow_d \ N\left(0, \Sigma_0\right),
    \end{split}
\end{align*}
    where $\Sigma_0 := \Cov_{P_0}(\phi_{0,1}(O))
    + \gamma^{-1} \Cov_{Q_0}(\phi_{0,2}(S; v_0))$.
\end{theorem}
To construct Wald-based confidence intervals (CIs) for $\psi_{0,0}$, we estimate the asymptotic covariance $\Sigma_0$ by plugging in consistent estimators of each component.
That is, we use consistent estimators $A_m$, $Z_m$, and $W_m$ of $A$, $Z$, and $W$, respectively.
Note that the estimators and CIs may be constructed using only the sampled subsets.
If $\psi_{0,0,j} = 0$ for any $j$, then the contribution from sampling observations to the asymptotic covariance term corresponding to index $j$ will be zero, leading to some additional complications.
We discuss this case further in the next section.

Conditions (A1)--(A4) are required to control the contribution from estimating $f_{0,s}$ for each $s \in \mathcal{S}$.
\citet{williamson2020b} show that these conditions are satisfied for $R^2$, deviance, accuracy, and AUC.
Conditions (B1)--(B2) place restrictions on the class of estimators of $f_{0,s}$ that we may consider.
While condition (B1) holds for many estimators (e.g., generalized additive models \citep{hastie1990}), we show in Section \ref{sec:sims} that this condition may only need to be approximately satisfied.
Condition (B2) is implied by a form of consistency of $f_{n,s}$.

Finally, condition (C1) is necessary to control the contribution from having had to estimate $Q_0$.
Because $\psi_{0,n}$ is an asymptotically efficient estimator of $\psi_{0,0}$, this condition implies that sampling $m = \Theta(n)$ subsets is asymptotically optimal, up to a constant factor proportional to $\gamma^{-1}$.
Intuitively, this is because there is an irremovable error contribution from having sampled $n$ training observations.
As such, we simply need to sample enough subsets for the second error term in \eqref{eq:decomposition} to be on the same order as the first term.
Moreover, because the distribution $Q_0$ places the heaviest weight on subset sizes at the extremes (closest in size to the empty set and full set), we do not need to estimate a large number of unique prediction functions in practice.
To our knowledge, this is the first result that delineates the number of feature subsets to sample for constructing an asymptotically normal estimator of Shapley values.

\subsection{Testing the null SPVIM hypothesis}\label{sec:testing}

We now use Theorem~\ref{thm:shapley_vim} to construct a test for the null hypothesis that a variable is not important, i.e., $\psi_{0,0,j} = 0$ for some $j$.
When a variable $X_j$ has null importance, the true value $\psi_{0,0,j}$ is at the boundary of the parameter space, and the contribution to the asymptotic variance from sampling observations in Theorem~\ref{thm:shapley_vim} is zero.
This may cause difficulties in hypothesis testing: as the number of sampled subsets grows, the contribution to the asymptotic variance from sampling subsets tends to zero.
Thus, in the limit, a hypothesis test based on the estimator of this asymptotic variance proposed in the previous section will fail to appropriately control the type I error.

Instead, we rely on sample-splitting to construct a valid test of the $\delta$-null hypothesis of the $j$th SPVIM value, i.e., $H_{0,j}: \psi_{0,0,j} \in [0, \delta]$.
In our approach, we make use of the fact that $\psi_{0,0,\emptyset}$ may be nonzero for some predictiveness measures (e.g., AUC).
Based on one portion of the data, construct estimator $\psi_{m,n,j,+} := \psi_{m,n,j} + \psi_{m,n,\emptyset}$ of $\psi_{0,0,j} + \psi_{0,0,\emptyset}$ and obtain an estimator $\sigma_{n,j}^2$ of the variance $\sigma_{0,j}^2 := (\Sigma_0)_{jj}$.
Based on the remaining data, obtain an estimator $\psi_{m,n,\emptyset,1}$ of $\psi_{0,0,\emptyset}$ with corresponding variance estimator $\sigma_{n,\emptyset}^2$.
Then, we calculate a test statistic $T_n := \frac{(\psi_{m,n,j,+} - \psi_{m,n,\emptyset,1}) - \delta}{\sqrt{n_1^{-1}\sigma_{n,j}^2 + 2*n_2^{-1}\sigma_{n,\emptyset}^2}}$ and its corresponding $p$-value $p_n := 1 - \Phi(T_n)$, where $n_1$ and $n_2$ denote the respective sample sizes of the split dataset and $\Phi$ denotes the standard normal cumulative distribution function.
We reject $H_0$ if and only if $p_n < \alpha$ for some pre-specified level $\alpha$.
Under conditions (A1)--(A4), (B1)--(B2), and (C1), for any $\alpha \in (0,1)$, the proposed test is consistent and has type I error equal to $\alpha$.

\section{Local and group variable importance}
Until now, we have focused on a global measure of importance by integrating over the entire distribution $P_0$.
For certain settings, we may be interested instead in a local version of variable importance.
A simple extension of \eqref{eq:class_shap} or \eqref{eq:wls} allows us to define a local version of variable importance: for a subpopulation $A \subseteq \mathcal{X}$,
\begin{align*}
    \psi_{0,0,j}(A) := \ \frac{1}{p}\sum_{s \in \mathcal{S}}\binom{p-1}{\lvert s \rvert}^{-1}&\{V(f_{0,s\cup j}, P_{0 \mid X \in A}) \\
    &- V(f_{0,s}, P_{0 \mid X \in A})\}, \notag
\end{align*}
where we have simply plugged the conditional distribution $P_{0 \mid X \in A}$ into \eqref{eq:class_shap}.
Taken to the extreme, where the subpopulation $A$ consists only of a single observation, this definition of local feature importance is equivalent to the SHAP values considered by \citet{lundberg2017}, though here we use an arbitrary measure of predictiveness in place of the conditional expectation.
Unfortunately, valid statistical inference on this individual-observation-level importance appears difficult, if not impossible.

In addition, if there is some scientifically meaningful partition of the features, we can extend SPVIM to these feature subgroups.
For example, one may group together all measurements from the same medical device.
Let the partition of features into groups be denoted $\mathcal{P} := \{s_1, \dots, s_k\}$ where $s_i \in \mathcal{S}$ and $\bigcup_{i = 1}^k s_i = N$, and $s_i \bigcap s_j = \emptyset$ for every $(i,j)$ pair.
Then the Shapley-based population variable \textit{group} importance measure may be determined as in \eqref{eq:class_shap}, where the sum is taken over all subsets in $\mathcal{P}$.

\section{Simulation study}
\label{sec:sims}

In this section, we present simulation results validating our statistical inference procedure for SPVIM in finite samples.
We consider 200 covariates $X \sim N_{200}(0, \Sigma)$. The variance-covariance matrix $\Sigma$ has diagonal equal to 1 and several correlated features: $Cov(X_1, X_{11}) = 0.7$; $Cov(X_3, X_{12}) = Cov(X_3, X_{13}) = 0.3$; and $Cov(X_5, X_{14}) = 0.05$. The covariance of the remaining feature pairs is zero. Based on these covariates, we observe a continuous outcome $Y \mid X = x \sim N(f(x), 1)$, where
\begin{align*}
    f(x) =& \ \sum_{j\in \{1,3,5\}} f_j(x_j), \\
    f_1(x) =& \ \text{sign}(x), \\
    f_3(x) =& \ (-6)I(x \leq -4) + (-4)I(-4 < x \leq -2) \\
    & + (-2)I(0 \leq x < -2) + 2I(2 < x \leq 4)\\
    & + 4I(x > 4), \text{ and } \\
    f_5(x) =& \ (-1)I(x \leq -4 \text{ or } -2 < x \leq 0 \text{ or } 2 < x \leq 4) \\
    &+ I(-4 < x \leq -2 \text{ or } 0 < x \leq 2 \text{ or } x > 4).
\end{align*}
In this data-generating mechanism, the vector $(X_1, X_3, X_5)$ is directly relevant to predicting the outcome, while the vector $(X_{11}, \ldots, X_{14})$ is only related to the outcome through correlation with $(X_1, X_3, X_5)$; the remaining 193 features are pure noise.
We generated 1,000 random datasets of size $n \in \{500, 1000, 2000, 3000, 4000\}$.
The true SPVIM values for predictiveness defined in terms of $R^2$ are approximately $(0.19, 0.29, 0.23, 0.04, 0.01, 0.01, 0)$ for the non-noise features, respectively, and zero for the remaining features.

To obtain each $f_{n,s}$ we fit boosted trees \citep{friedman2001} using the Python package \texttt{xgboost} \citep{chen2016} with maximum tree depth equal to one, learning rate equal to $10^{-2}$, and $\ell_2$-regularization parameter equal to zero.
The number of trees varied among $\{50, 100, 250, 500, 1000, \ldots, 3000\}$ and the $\ell_1$-regularization parameter varied among $\{10^{-3}, 10^{-2}, 0.1, 1, 5, 10\}$; the combination of these parameters was tuned using five-fold cross-validation to minimize the mean squared error (MSE).

We computed the relevant SPVIM estimator using Algorithm~\ref{alg:spvi_cls}, where we sampled $m = 2n$ subsets and estimated predictiveness using five-fold cross-fitting.
For comparison, we computed the mean absolute SHAP value \citep{lundberg2017}, where the average was taken over all observations.
This allows us to directly evaluate the accuracy of algorithmic VIMs for estimating the population VIMs.
We then computed the empirical MSE scaled by $n$, the empirical coverage of nominal 95\% CIs, and the empirical power of our proposed hypothesis test.
Finally, we compare the accuracy of our SPVIM estimates and the mean SHAP values in terms of their correlation with the true SPVIM values.
All analyses were performed on a computer cluster with 32-core CPU nodes with 64 GB RAM.

We display the results of this experiment in Figure~\ref{fig:shapley_sim}.
We see that as $n$ increases, the scaled empirical MSE of our estimator decreases to a fixed level --- namely, the scaled empirical variance --- for each feature.
This matches our expectations from Section~\ref{sec:inf}: the scaled empirical bias of our proposed estimator should tend to zero with increasing sample size, while the scaled empirical variance tends to the asymptotic variance.
We note here that while boosted trees are a popular estimation procedure, they do not necessarily satisfy condition (B1) \citep[see, e.g.,][]{zhang2005}.
Thus, the convergence observed here provides some empirical evidence that condition (B1) may only need to hold approximately in practice.
We also find that the coverage of nominal 95\% confidence intervals increases to the nominal level as the sample size increases.
Our proposed hypothesis test controls the type I error rate and is consistent: the empirical type I error rate is at the nominal level for null feature $X_6$, while the empirical power is near one for each of the directly important features.
Power tends to be small for the indirectly important features $(X_{11}, \ldots, X_{14})$, especially at small sample sizes; this reflects the fact that the importance of these features is closer to the null hypothesis than the importance of the directly relevant features.
Finally, we see that SPVIM estimates are more correlated with the true population importance than SHAP values.
We provide the estimated SPVIM and mean absolute SHAP values in the Supplement.

\begin{figure*}
   \hspace{-0.55cm}\includegraphics[width = 1.1\textwidth]{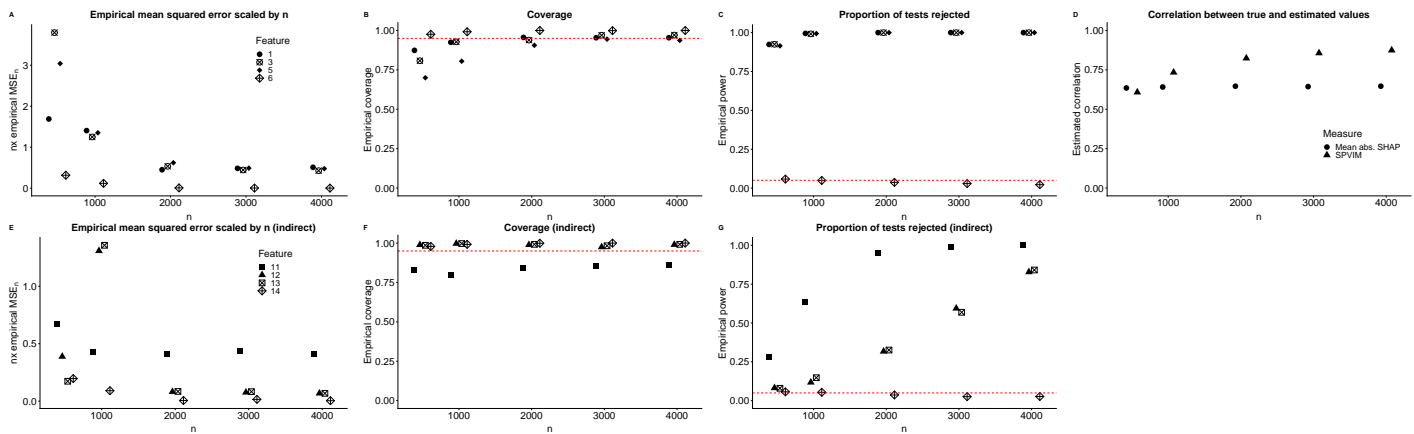}
    \caption{
    Performance of our statistical inference procedure for estimating the Shapley-based population variable importance (SPVIM) with respect to $R^2$ using $n$ training observations and $2n$ sampled subsets.
    (A, E) Empirical MSE for the proposed plug-in estimator scaled by $n$ for $j \in \{1,3,5,6\}$ and $j \in \{11, 12, 13, 14\}$, respectively;
    (B, F) Empirical coverage of nominal 95\% confidence intervals;
    (C, G) Empirical power of the hypothesis testing procedure for null hypothesis that the $j$th variable has null importance;
    (D) Kendall's tau between the true and estimated SPVIM values using our approach versus the mean absolute SHAP value.
    }
    \label{fig:shapley_sim}
\end{figure*}

\section{Predicting mortality of patients in the intensive care unit}
\label{sec:data}

We now analyze data on patients' stays in the ICU from the Multiparameter Intelligent Monitoring in Intensive Care II (MIMIC-II) database \citep{silva2012predicting}.
We consider 4000 records on several features: five general descriptors collected upon admission to the ICU, and 15 features --- including the Glasgow Coma Scale (GCS), blood urea nitrogen (BUN), and heart rate --- measured over the course of the first 48 hours after admission to the ICU.
The outcome of interest is in-hospital mortality.
Rather than use the entire time series, we simplify the analysis by first computing the minimum, average, and maximum value of each of the time-series features used in the simplified acute physiology (SAPS) I or II scores.
The SAPS scores are established measures for estimating the mortality risk of ICU patients.
We then remove any features that are measured in fewer than 70\% of the patients.
When combined with the general descriptor variables, a total of 37 extracted features remain.
We provide a full list of these extracted features in the Supplement.

We estimate the SPVIM for each variable using AUC to measure predictiveness.
For comparison, we also provide the mean absolute SHAP value obtained from Tree SHAP \citep{lundberg2020} and Kernel SHAP \citep{lundberg2017}; and the proportion of times a feature was selected across test instances using LIME \citep{ribeiro2016}.
We discuss conditions under which the mean absolute SHAP value is a suitable proxy for the SPVIM in Section~\ref{sec:shap_v_spvim} in the Supplement.

We obtained estimates of each $f_{0,s}$ using two separate procedures.
In the first analysis, we maximized the empirical log likelihood using boosted trees with maximum depth equal to four, learning rate equal to $10^{-3}$, and a number of estimators in $\{2000, 4000, \dots, 12000\}$ selected using five-fold cross-validation.
In the second analysis, we maximized the empirical log likelihood by fitting ensembles of five dense ReLU neural networks (NNs) with architectures chosen from $\{(37,25,25,20,10,1), (37,25,20,1), (37,25,20,20,1)\}$ using 5-fold cross-validation.
The NNs were trained using Adam \citep{kingma2014adam} with a maximum of 2000 iterations and with $\ell_2$ regularization parameter equal to 0.1.
We again used 5-fold cross-fitting to estimate the predictiveness measures for the sampled subsets.
Using our procedure, we fit models for only 119 unique subsets and computed SPVIM estimates in two hours for each analysis.
LIME had similar computation time (1.7 hours) in the case of NNs, but longer computation time (4 hours) in the case of trees.
The computation time of both our procedure and LIME falls between the highly specialized Tree SHAP algorithm, which completed in a few minutes, and the general-purpose Kernel SHAP, which took approximately 20 hours.

\begin{figure*}
    \centering
    \includegraphics[width = 1\textwidth]{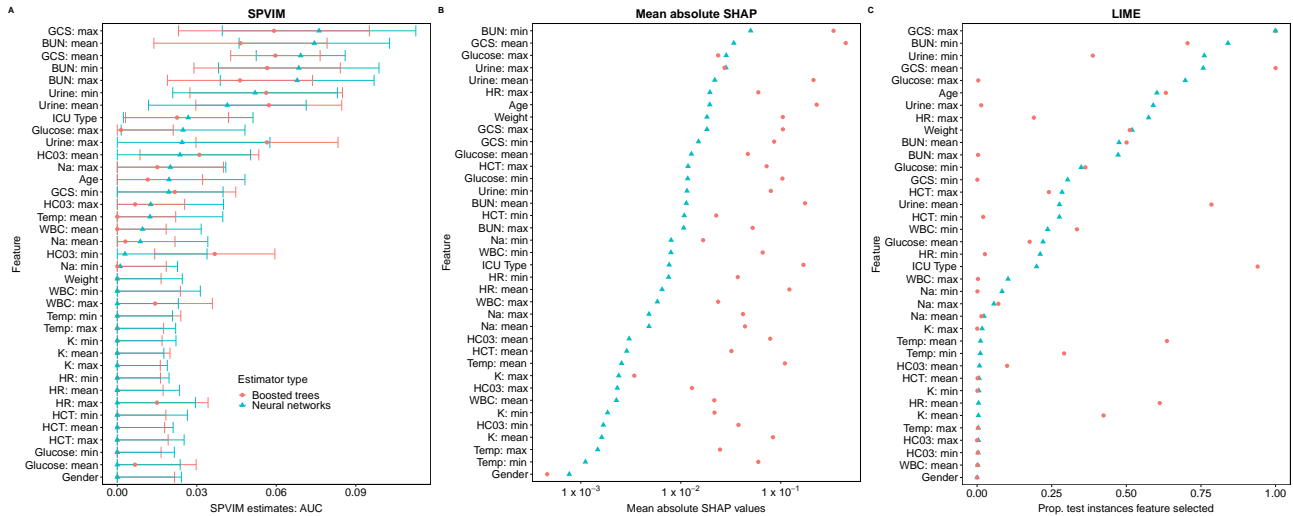}
    \caption{
	    We estimated importance of features for predicting in-hospital death in the ICU using our statistical inference procedure for SPVIM with respect to AUC (A), the mean absolute SHAP value (B), and LIME (C).
	    Red circles and green triangles denote estimates from fitting boosted trees and neural networks, respectively.
        The features are ordered from top to bottom by their point estimate from the neural networks procedure.
	    95\% confidence intervals only appear in (A) since there is no statistical inference procedure for mean absolute SHAP values or LIME.}
    \label{fig:icu_results}
\end{figure*}

In Figure~\ref{fig:icu_results}, we display the estimates from each VIM and both estimation procedures.
We first focus on the SPVIM estimates provided in Panel A.
The GCS is estimated to be the most important feature using both trees and NNs, though different summaries of the GCS are most important across the two procedures (mean for trees and max for NNs).
This result matches prior knowledge: GCS is used to assess the level of consciousness of patients and is the highest scoring item in the SAPS scores.
We find that the confidence intervals for SPVIM are quite wide, which is important information for placing the results in context.

Next, we compare the agreement between rankings calculated based on the fitted boosted trees and NNs for the SPVIM estimates, mean absolute SHAP values (Figure~\ref{fig:icu_results} panel B), and LIME (panel C).
There is considerably more agreement between the two procedures for the SPVIM estimates than for the SHAP value estimates and LIME proportions.
The estimated Kendall's tau between procedures is 0.71 for our SPVIM estimator vs 0.37 for SHAP and 0.39 for LIME.
Given the large discrepancies between the algorithmic VIMs, we conclude that they are poor proxies for our population VIM.
Instead, one should use a procedure specifically designed to estimate SPVIM.

Finally, we find that the feature rankings within trees or NNs from our SPVIM estimator, SHAP, and LIME are substantively different.
One noticeable difference is that SHAP and LIME values for several summary statistics derived from the same measurement (e.g., min, mean, and max GCS) differ widely; this should not occur, since these summary statistics are highly correlated.
On the other hand, SPVIM estimates for summary statistics derived from the same measurement tend to be more similar.

\section{Discussion}

We have proposed a computationally tractable statistical inference procedure for the Shapley population variable importance measure (SPVIM).
Methods for estimating SPVIM are complementary to those for estimating algorithmic variable importance.
The former helps us understand the underlying data-generating mechanism and can help guide future experiments; the latter helps us interpret a particular fitted model.
Here, we define SPVIM with respect to an arbitrary measure of predictiveness, allowing the data analyst to select the most appropriate measure for the task at hand.
Since the SPVIM is also defined relative to the population, the target of inference is not affected by the choice of prediction algorithm.
We have derived the asymptotic distribution of an SPVIM estimator based on randomly sampled feature subsets, and have used this distribution to construct asymptotically normal point estimates, valid confidence intervals, and hypothesis tests with the correct type I error control.
Notably, we determined a minimum number of feature subsets to sample: we show that our estimator only needs to fit prediction models for $m = \Theta(n)$ sampled subsets for its error to be on the same order as an estimator that fits prediction models for all possible subsets.

In this manuscript, we have focused on quantifying the importance of a variable averaging across the entire population.
Local importance measures can be obtained by restricting to smaller subpopulations.
However, as the subpopulations decrease in size, the uncertainty of our estimates increases.
Our asymptotic results do not apply to the most extreme case, the variable importance at the level of a single observation.
Nevertheless, this value may be of interest in some tasks.
Further work is necessary to define relevant importance measures at the single-observation-level and derive procedures with the desired performance.

Finally, we caution against interpreting SPVIM estimates in a causal manner.
SPVIM reflects importance in the oracle prediction model rather than importance in the oracle causal model.
In many scientific applications, the importance in the causal model is of ultimate interest.
To get causal importance, one may need to employ techniques from causal inference.
Recent developments relating prediction models and causal models may also be of use in these cases \citep{Arjovsky2019-cw}.

\section{Acknowledgments}

The authors wish to thank Jessica Perry, Noah Simon, the anonymous reviewers, and the meta-reviewer for insightful comments that improved this manuscript. BDW was supported by NIH award F31 AI140836. The opinions expressed in this manuscript are those of the authors and do not necessarily represent the official views of the NIAID or the NIH.

\clearpage
\newpage

\section*{Supplementary Material}\label{sec:supp}

\section{Proof of Theorem 1}\label{sec:proofs}

We consider the decomposition
\begin{align}\label{eq:decomposition2}
   \psi_{m,n} - \psi_{0,0} =& \ (\psi_{0,n} - \psi_{0,0}) + (\psi_{m,0} - \psi_{0,0}) + r_{m,n},
\end{align}
where
\begin{align}
\begin{split}
\psi_{m,0}
= & \argmin_{\psi \in \mathbb{R}^{p+1}} E_{Q_m} \left[
(Z(S) \psi - v_{0,S})^2
\right] \\
& \text{subject to } G\psi = c_0
\end{split}
\label{eq:plugin_shap_recall}
\end{align}
and $r_{m,n} := (\psi_{m,n} - \psi_{m,0}) - (\psi_{0,n} - \psi_{0,0})$.

We first control the first term in \eqref{eq:decomposition2}.
Since Shapley values are defined as a linear combination of the predictiveness vector, let the matrix $B(p)$ encode these weights.
Note that this matrix only depends on $p$.
The first row of $B(p)$, denoted $[B(p)]_1$, is given by $[B(p)]_1 = z(\emptyset)$.
The matrix entry in row $j = 2, \dots, p+1$ and column $i = 1, \dots, 2^p$ is
$$
[B(p)]_{ji} := \frac{1}{p}(-1)^{I\{(j-1) \notin s_i\}}\binom{p-1}{\lvert s_i \rvert - I\{(j-1) \in s_i\}}^{-1},
$$
where column $i$ corresponds to subset $s_{(i)}$.
Then
\begin{align*}
    \psi_{0,0} &= B(p) v_0 \text{ and }\\
    \psi_{0,n} &:= B(p) v_n.
\end{align*}
Under the collection of conditions implied by (A1)--(A4) and (B1)--(B2) for each subset $s \in \mathcal{S}$, a straightforward application of the functional delta method and Theorem 2 of \citet{williamson2020b} yields that $\psi_{0,n}$ is an asymptotically linear estimator of $\psi_{0,0}$ with influence function given by
\begin{align}\label{eq:equation_1}
    \phi_{0,1}: o \mapsto B(p)\dot{V}_0(o),
\end{align}
where $\dot{V}_0$ is the influence function of $v_n$ and is defined in the main manuscript.

We now control the second term in \eqref{eq:decomposition2}.
We use the equivalent weighted least squares formulation of the Shapley values,
\begin{align}
\psi_{0,0} &= \argmin_{\psi: G\psi = c_0 } E_{Q_0} \left(
Z(S) \psi - v_{0,S}
\right)^2
\label{eq:constrained_wls_1} \text{ and }
\\
\psi_{m,0} &= \argmin_{\psi: G\psi = c_0 } E_{Q_m} \left(
Z(S) \psi - v_{0,S}
\right)^2.
\label{eq:constrained_wls_2}
\end{align}
We write the the QR decomposition of $G^\top$ as
\begin{align*}
G^\top = U
\left[
\begin{matrix}
R\\
0
\end{matrix}
\right]
=
\left[
\begin{matrix}
U_1 & U_2
\end{matrix}
\right]
\left[
\begin{matrix}
R\\
0
\end{matrix}
\right],
\end{align*}
where $U$ is an orthonormal matrix and $R$ is an upper triangular matrix.
$U_1$ is a 2-column orthogonal matrix corresponding to the column space of $G^\top$ and $U_2$ is a $(p - 1)$-column orthogonal matrix corresponding to its null space.
As such, we can reparameterize the constrained least squares problems in \eqref{eq:constrained_wls_1} and \eqref{eq:constrained_wls_2} using the vector $\theta \in \mathbb{R}^{p + 1}$, where $\psi = U\theta$.
The constraint $G\psi = c_0$ implies that
\begin{align}
\left[
\begin{matrix}
R^\top & 0
\end{matrix}
\right]
 \theta
 =
 R^\top \theta_1
 = c_0,
\label{eq:theta1_solve}
\end{align}
where $\theta_1$ is the first 2 elements of $\theta$.
Thus $\theta_1$ is fixed by the constraint, while $\theta_2$ is not constrained.
This implies that the solutions to \eqref{eq:constrained_wls_1} and \eqref{eq:constrained_wls_2} correspond to $\theta$ with $\theta_1$ as the solution to \eqref{eq:theta1_solve} and $\theta_2$ as the solution to the unconstrained least squares problems
\begin{align*}
\theta_{2,0} &= \argmin_{\theta_2 \in \mathbb{R}^{p - 1}} E_{Q_0} \left(
Z(S) (U_1 \theta_1 + U_2 \theta_2) - v_{0,S}
\right)^2 \text{ and }
\\
\theta_{2,m} &= \argmin_{\theta_2 \in \mathbb{R}^{p - 1}} E_{Q_m} \left(
Z(S) (U_1 \theta_1 + U_2 \theta_2) - v_{0,S}
\right)^2.
\end{align*}

A straightforward application of Theorem 5.23 in \citet{vandervaart2000} yields that $\theta_{2,m}$ is an asymptotically linear estimator of $\theta_{2,0}$, with
\begin{align*}
&\sqrt{m}(\theta_{2,m} - \theta_{2,0}) \notag \\
=& \ -\frac{1}{\sqrt{m}}\sum_{j=1}^m
\bigg[V^{-1}
\left\{
z(S_j)^\top
(U_1 \theta_{1} + U_2 \theta_{2,0})
- v_{0, S_j}
\right\} \notag \\
&\hspace{.8in}\times U_2^\top z(S_j)\bigg] + o_P(1),
\end{align*}
where $V = U_2^\top Z^\top W Z U_2 $.
Thus, $\psi_{m,0}$ is an asymptotically linear estimator of $\psi_{0,0}$, i.e.,
\begin{align}
\sqrt{m}(\psi_{m,0} - \psi_{0,0})
= \ \frac{1}{\sqrt{m}}\sum_{j=1}^m
\phi_{0,2}(S_j; v_0) + o_P(1)
\label{eq:equation_2}
\end{align}
where $\phi_{0,2}(S; v_0)$ is defined as
\begin{align*}
\phi_{0,2}: s \mapsto - U_2
V^{-1}
\left[
z(s)^\top
\psi_{0,0}
- v_{0, s}
\right]
U_2^\top z(s).
\end{align*}

Finally, we control the remainder term $r_{m,n}$.
By the KKT conditions in the main manuscript, we have that
\begin{align*}
\psi_{m,n} & = C_2(Q_m) v_n
\end{align*}
where $C_2(Q_m)$ is defined as
\begin{align*}
\left[
\begin{matrix}
I_{p + 1} & 0
\end{matrix}
\right]
\left[
\begin{matrix}
2 Z_m^\top W_m Z_m & G^\top\\
G & 0
\end{matrix}
\right]
^{-1}
\left[
\begin{matrix}
2 \sqrt{W_m}\\
e_{\emptyset}\\
e_{N} - e_{\emptyset}
\end{matrix}
\right]
\end{align*}
and $e_{s} \in \{0,1\}^{2^p}$ is a one-hot vector for the set $s$.
Likewise, define $C_2(Q_0)$ as
\begin{align*}
\left[
\begin{matrix}
I_{p + 1} & 0
\end{matrix}
\right]
\left[
\begin{matrix}
2 Z^\top W Z & G^\top\\
G & 0
\end{matrix}
\right]
^{-1}
\left[
\begin{matrix}
2 \sqrt{W}\\
e_{\emptyset}\\
e_{N} - e_{\emptyset}
\end{matrix}
\right].
\end{align*}
Then
\begin{align*}
r_{m,n} = & \ (\psi_{m,n} - \psi_{m,0}) - (\psi_{0,n} - \psi_{0,0})\\
= & \ \{C_2(Q_m) - C_2(Q_0)\}(v_n - v_0).
\end{align*}
Since the empirical distribution $Q_m$ converges weakly to $Q_0$, then $C_2(Q_m) \to_p C_2(Q_0)$.
Moreover, if (A1)--(A4) and (B1)--(B2) hold for each subset $s \in \mathcal{S}$, then $v_n - v_0 = O_p(n^{-1/2})$.
Thus
\begin{align}\label{eq:remainder}
r_{m,n} =& \ o_P(n^{-1/2}).
\end{align}

In view of \eqref{eq:equation_1}, \eqref{eq:equation_2}, and \eqref{eq:remainder}, we can write
\begin{align*}
& \sqrt{n}(\psi_{m,n} - \psi_{0,0}) \notag \\
=& \ \sqrt{n}(\psi_{0,n} - \psi_{0,0}) + \sqrt{n}(\psi_{m,0} - \psi_{0,0}) + \sqrt{n}r_{m,n} \notag \\
=& \ \frac{1}{\sqrt{n}}\sum_{i=1}^n \phi_{0,1}(O_i)
 + \frac{1}{\sqrt{n \gamma_n }} \sum_{i=1}^{n \gamma_n} \phi_{0,2}(S_i; v_0)
 + o_P(1).
\end{align*}
Because $O$ and $S$ are sampled independently and $\gamma_n \rightarrow_p \gamma$, then, by Slutsky's theorem, we have that
\begin{align*}
&\sqrt{n}(\psi_{m,n} - \psi_{0,0})
\to_d \\
&N\left[
0,
\Cov\{\phi_{0,1}(O)\}
+ \gamma^{-1} \Cov\{\phi_{0,2}(S; v_0)\}
\right]
\end{align*}
Finally, note that if $\gamma_n \to \infty$, then the second term in the asymptotic variance is zero.

\section{Additional technical details}

\subsection{Shapley values minimize a weighted least squares problem}

Recall the classical Shapley formula: for $j = 1, \ldots, p$,
\begin{align*}
    \psi_{0,j} = \frac{1}{p}\sum_{s \in N \setminus \{j\}} \binom{p-1}{\lvert s \rvert}^{-1}(v_{0, s\cup \{j\}} - v_{0,s}).
\end{align*}
Our goal is to show that the solution $x^*$ to the minimization problem
\begin{align*}
      &\minimize_{x \in \mathbb{R}^{p+1}} \frac{1}{2}\lVert \sqrt{W}(Zx - v_0)\rVert_2^2 \\
      &\text{subject to } \sum_{j=1}^p x_j = v_{0,N} - v_{0,\emptyset} \text{ and } x_0 = v_{0,\emptyset} \notag
\end{align*}
satisfies $x^*_j = \psi_{0,j}$ for $j = 1, \ldots, p$.

Since the classical Shapley values in the first display and the solution to the constrained, weighted least squares problem are both linear in $v_0$, if we can prove that the two values are equivalent for all one-hot vectors $v_{(k)}$ for $k = 1, \dots, 2^p$, then we will have proved that the two values are equal.
Our first result provides the form of the classical Shapley values for a one-hot vector. As in the main manuscript, $s_{(k)}$ refers to the $k$th ordered subset of $N = \{1, \dots, p\}$.
\begin{lemma}\label{lem:unit_class_shap}
    For $j = 1, \ldots, p$, the classical Shapley value corresponding to one-hot vector $v_{(k)}$ is given by
    \begin{align*}
        \psi_j(v_{(k)}) =& \frac{1}{p}1\{k \neq 1, k \neq 2^p\}\\
        &\times \bigg[\binom{p-1}{\lvert s_{(k)} \rvert - 1}^{-1}1\{j \in s_{(k)}\}\\
        &\hspace{.2in}- \binom{p-1}{\lvert s_{(k)} \rvert}^{-1}1\{j \notin s_{(k)}\}\bigg] \\
        &+ \frac{1}{p}(v_{(k),N} - v_{(k),\emptyset}).
    \end{align*}
\end{lemma}
\begin{proof}
    For $j = 1, \dots, p$, the classical Shapley formula states that
    \begin{align*}
        \psi_j(v_{(k)}) = \frac{1}{p}\sum_{S \subseteq N \setminus \{j\}}\binom{p-1}{\lvert S \rvert}^{-1}(v_{(k), s \cup \{j\}} - v_{(k),s}).
    \end{align*}
    Since $s_{(k)}$ corresponds to $v_{(k)}$, we have that
    \begin{align*}
        \psi_j(v_{(k)}) = \begin{cases}
                            \frac{1}{p}\binom{p-1}{\lvert s_{(k)} \rvert - 1}^{-1}(1 - 0) & \text{ if } j \in s_{(k)} \\
                            \frac{1}{p}\binom{p-1}{\lvert s_{(k)} \rvert}^{-1}(0 - 1) & \text{ if } j \notin s_{(k)}
                          \end{cases}.
    \end{align*}
    For $k = 2, \dots, 2^p - 1$, we have that $v_{(k),N} - v_{(k),\emptyset} = 0$. Thus, the claim is proved for these values of $k$.
    Note that $\psi_j(v_{(1)}) = -1/p$, while for $k = 2^p$, $\psi_j(v_{(2^p)}) = 1/p$ by the definition above.
    Thus, the claim is proved for all $k$.
\end{proof}
Our next result provides the solution to the constrained, weighted least squares problem
\begin{align}\label{eq:unit_wls}
      &\minimize_{x \in \mathbb{R}^{p+1}} \lVert \sqrt{W}(Zx - v_{(k)})\rVert_2^2 \\
      &\text{subject to } \sum_{j=1}^p x_j = v_{(k),N} - v_{(k),\emptyset} \text{ and } x_0 = v_{(k),\emptyset}. \notag
\end{align}
\begin{lemma}\label{lem:unit_wls}
    For $j = 1, \dots, p$, the solution to \eqref{eq:unit_wls} is given by
    \begin{align*}
        x^*_j(v_{(k)}) =& \frac{1}{p}1\{k \neq 1, k \neq 2^p\} \\
        &\times \bigg[\binom{p-1}{\lvert s_{(k)} \rvert - 1}^{-1}1\{j \in s_{(k)}\}\\
        &\hspace{.2in}- \binom{p-1}{\lvert s_{(k)} \rvert}^{-1}1\{j \notin s_{(k)}\}\bigg] \\
        &+ \frac{1}{p}(v_{(k),N} - v_{(k),\emptyset}).
    \end{align*}
\end{lemma}
\begin{proof}
    For ease of notation, we use $x^*_j$ and $x^*_j(v_{(k)})$ interchangeably.
    Consider the Lagrangian of \eqref{eq:unit_wls}, given by
    \begin{align*}
        L(v_{(k)}, x, \lambda) =& \ \lVert \sqrt{W}(Zx - v_{(k)}) \rVert_2^2 + \lambda^\top (Gx - v_c),
    \end{align*}
    where $G = \begin{bmatrix} z(\emptyset) \\ z(N) - z(\emptyset) \end{bmatrix} = \begin{bmatrix} 1 & 0 & \ldots & 0 \\ 0 & 1 & \ldots & 1 \end{bmatrix}$ and $v_c = \begin{bmatrix} v_{(k),\emptyset} \\ v_{(k),N} - v_{(k),\emptyset}\end{bmatrix}$.
    Setting the gradient of the Lagrangian equal to zero, we find that $x^*$ must satisfy
    \begin{align*}
        \nabla_x L(v_{(k)}, x, \lambda) =& \ Z^\top W(Z x - v_{(k)}) + G^\top \lambda \stackrel{\text{set}}{=} 0 \\
        \Rightarrow 0 =& \ Z^\top W(Zx^* - v_{(k)}) + G^\top \lambda^* \\
        \nabla_\lambda L(v_{(k)}, x, \lambda) =& \ Gx - v_c \stackrel{\text{set}}{=} 0 \\
        \Rightarrow 0 =& \ Gx^* - v_c.
    \end{align*}
    This yields that
    \begin{align}\label{eq:unit_kkt}
        Z^\top W Z x^* = Z^\top W v_{(k)} - G^\top \lambda^*.
    \end{align}
    Note that $Z^\top W v_{(k)} = w_{s_{(k)}}z(s_{(k)})$, where $w_{s_{(k)}}$ is the weight for subset $s_{(k)}$, and for ease of notation we set $w_S = \binom{p-2}{\lvert S \rvert - 1}^{-1}$, with $w_{\emptyset} = 1$.
    We now find the value of $\lambda^*$.
    We denote the index of the first row of $Z^\top W Z x^*$ by zero, to match with $x^*_0$.
    Expanding the matrix notation in \eqref{eq:unit_kkt}, the first row of \eqref{eq:unit_kkt} states that
    \begin{align*}
        w_{s_{(k)}} - \lambda_1^* =& \ [Z^\top W Z x^*]_0 \\
        =& \ \left(\sum_{S \in \mathcal{S}} w_{S}\right)x_0^* + \sum_{j=1}^p \left(\sum_{S \in \mathcal{S} \, : \, j \in S} w_S\right)x_j^* \\
        =& \ \left(\sum_{S \in \mathcal{S}} w_{S}\right)v_{(k),\emptyset} \\
        &+ \left(\sum_{S \in \mathcal{S} \, : \, 1 \in S} w_S\right)(v_{(k),N} - v_{(k),\emptyset}),
    \end{align*}
    where we have made use of the constraints from \eqref{eq:unit_wls} and the symmetry of the weights.
    Thus,
    \begin{align*}
        \lambda^*_1 =& \ w_{s_{(k)}} - \left(\sum_{S \in \mathcal{S}} w_{S}\right)v_{(k),\emptyset} \\
        &\hspace{.2in}- \left(\sum_{S \in \mathcal{S} \, : \, 1 \in S} w_S\right)(v_{(k),N} - v_{(k),\emptyset}).
    \end{align*}
    For row $\ell = 1, \dots, p$, we have that
    \begin{align*}
        [Z^\top W Z x^*]_\ell =& \ \sum_{i=1}^{2^p}1\{\ell \in s(i)\}w_{s(i)}x^*_0 \\
        &+ \sum_{i=1}^{2^p}1\{\ell \in s(i)\}w_{s(i)}\sum_{j=1}^px_j^*1\{j \in s(i)\}) \\
        =& \ \left(\sum_{S\,:\, 1 \in S} w_S\right)x_0^* \\
        &+ \left(\sum_{S\,:\,1,2\in S}w_S\right)(v_{(k),N} - v_{(k),\emptyset}) \\
        &+ \left(\sum_{S\,:\,1\in S}w_S - \sum_{S\,:\,1,2\in S}w_S\right)x_\ell^*,
    \end{align*}
    using the symmetry of the weights.
    Thus, row $\ell$ of \eqref{eq:unit_kkt} is
    \begin{align}
        &w_{s_{(k)}}1\{\ell \in s_{(k)}\} - [G^\top \lambda^*]_\ell = [Z^\top W Z x^*]_\ell \notag\\
        \Rightarrow &w_{s_{(k)}}1\{\ell \in s_{(k)}\} - \lambda_2^* = \notag\\
        &\hspace{.2in}\left(\sum_{S\,:\, 1 \in S} w_S\right)x_0^* \notag\\
        &\hspace{.2in}+ \left(\sum_{S\,:\,1,2\in S}w_S\right)(v_{(k),N} - v_{(k),\emptyset}) \notag\\
        &\hspace{.2in}+ \left(\sum_{S\,:\,1\in S}w_S - \sum_{S\,:\,1,2\in S}w_S\right)x_\ell^*. \label{eq:row_ell}
    \end{align}
    Summing \eqref{eq:row_ell} across $\ell = 1, \dots, p$ yields
    \begin{align*}
        \lambda_2^* =& \ \frac{1}{p}\sum_{\ell = 1}^p\bigg[w_{s_{(k)}}1\{\ell \in s_{(k)}\} - \left(\sum_{S\,:\, 1 \in S} w_S\right)x_0^* \notag \\
        &\hspace{.5in}- \left(\sum_{S\,:\,1,2\in S}w_S\right)(v_{(k),N} - v_{(k),\emptyset}) \notag \\
        &\hspace{.5in}-\left(\sum_{S\,:\,1\in S}w_S - \sum_{S\,:\,1,2\in S}w_S\right)x_\ell^* \bigg]\notag \\
        =& \ \frac{1}{p}w_{s_{(k)}}\lvert s_{(k)}\vert - \left(\sum_{S\,:\, 1 \in S} w_S\right)v_{(k),\emptyset} \\
        &- \left(\sum_{S\,:\,1,2\in S}w_S + \frac{p-1}{p}\right)(v_{(k),N} - v_{(k),\emptyset}),
    \end{align*}
    where we have again made use of the constraints and the symmetry of $W$, as well as the difference-of-weights result that $\sum_{S\,:\,1\in S}w_S - \sum_{S\,:\,1,2\in S}w_S = (p-1)$.
    Plugging this result into \eqref{eq:row_ell} and rearranging terms yields that, for each $\ell = 1, \dots, p$,
    \begin{align}
        x_\ell^* =& \ \frac{w_{s_{(k)}}}{p-1}\left\{1\{\ell \in s_{(k)}\} - \frac{1}{p}\lvert s_{(k)} \rvert\right\} + \frac{1}{p}(v_{(k),N} - v_{(k),\emptyset}), \label{eq:x_ell}
    \end{align}
    where we have again made use of the constraints, the symmetry of $W$, and the difference-of-weights result.

    Note that for $k = 2, \dots, 2^p - 1$, $v_{(k),N} = v_{(k),\emptyset} = 0$.
    Thus, for $k = 2, \dots, 2^p - 1$, and $\ell = 1, \ldots, p$, if $\ell \in s_{(k)}$ then
    \begin{align*}
        x_\ell^* =& \ \frac{w_{s_{(k)}}}{p-1}\left\{1 - \frac{1}{p}\lvert s_{(k)} \rvert\right\} = \frac{1}{p}\binom{p-1}{\lvert s_{(k)} \rvert - 1}^{-1};
    \end{align*}
    if $\ell \notin s_{(k)}$ then
    \begin{align*}
        x_\ell^* =& \ \frac{w_{s_{(k)}}}{p-1}\left\{- \frac{1}{p}\lvert s_{(k)} \rvert\right\} = -\frac{1}{p}\binom{p-1}{\lvert s_{(k)} \rvert}^{-1}.
    \end{align*}
    Also, \eqref{eq:x_ell} implies that if $k = 1$ then $x_\ell^* = -\frac{1}{p}$, while if $k = 2^p$ then $x_\ell^* = \frac{1}{p}$.
    Thus,
    \begin{align*}
        x_\ell^*(v_{(k)}) =& \ \frac{1}{p}1\{k\neq 1, k\neq 2^p\}\\
        &\times \bigg[\binom{p-1}{\lvert s_{(k)} \rvert - 1}^{-1}1\{\ell \in s_{(k)}\}\\
        &\hspace{.2in}- \binom{p-1}{\lvert s_{(k)} \rvert}^{-1}1\{\ell \notin s_{(k)}\}\bigg] \\
        &+ \frac{1}{p}(v_{(k),N} - v_{(k),\emptyset}),
    \end{align*}
    precisely what we wished to show.
\end{proof}
Combining the results of Lemma~\ref{lem:unit_class_shap} and Lemma~\ref{lem:unit_wls}, we have that $x_j^*(v_{(k)}) = \psi_j(v_{(k)})$ for all one-hot vectors $v_{(k)}$, $k = 1, \dots, 2^p$.
Thus, the Shapley values are equivalent to the solution of the weighted least squares problem.

\subsection{SHAP values versus SPVIM}\label{sec:shap_v_spvim}
Under certain conditions, the mean absolute SHAP value is related to the SPVIM value.
Recall that for each feature subset $s \subseteq N$ and corresponding fitted models $\hat{f}_{s}$, the SHAP value for the $j$th feature at $x$ is defined as
$$
\sum_{s \in N \setminus \{j\}}\frac{1}{p} \binom{p-1}{\lvert s \rvert}^{-1}\{\hat{f}_{s\cup j}(x) - \hat{f}_{s}(x) \}.
$$
Suppose there exists a factor $c> 0$ such that for all feature subsets $s$, the scaled norm between oracle prediction models $f_{0, s \cup j}$ and $f_{0, s }$ provides a lower bound on the difference between their predictiveness measures, i.e.,
\begin{align}
\|f_{0,s  \cup j } - f_{0,s}\|_1 \lesssim c \left ( V(f_{0, s \cup j }, P_0) - V(f_{0, s}, P_0) \right ).
\label{eq:lower_bound}
\end{align}
Then it is easy to show that the mean absolute SHAP value for the oracle model implies large SPVIM values.
The lower bound \eqref{eq:lower_bound} holds if the predictiveness measure $V$ is convex in its first argument, such as when $V$ is the mean squared error.

\section{Additional numerical results}

In the main manuscript, we ran a 200-variable simulation with a continuous outcome. In Figure \ref{fig:add_sim_res}, we provide the estimated SPVIM value and mean absolute SHAP value for each sample size and feature considered in that analysis. The vertical bars denote the Monte-Carlo error based on 1000 replicates of the experiment for each sample size.

\begin{figure*}
    \centering
    \includegraphics[width = 0.8\textwidth]{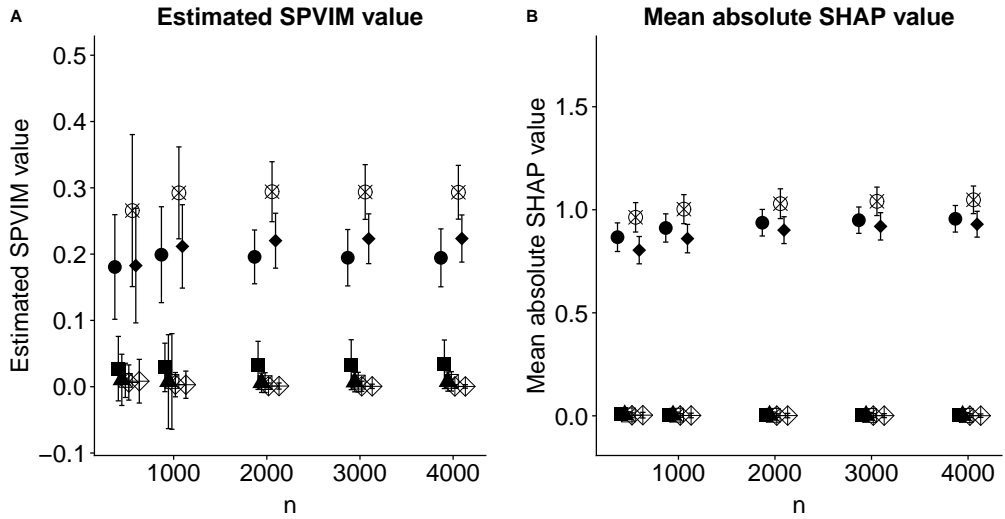}
    \caption{Average estimated SPVIM values (left) and mean absolute SHAP values (right) with Monte-Carlo error bars for the two hundred-variable simulation conducted in the main manuscript.
	Filled and crossed circles and filled and crossed diamonds denote $X_1$, $X_3$, $X_5$, and $X_6$ and $X_{14}$, respectively; filled squares, filled triangles, and crossed squares denote $X_{11}$, $X_{12}$, and $X_{13}$, respectively.
	The true SPVIM values are approximately $(0.192, 0.291, 0.228, 0.037, 0.01, 0.01, 0)$ for the non-noise features (1, 3, 5, 11, 12, 13, 14) and zero for $X_6$.}
    \label{fig:add_sim_res}
\end{figure*}

\section{Additional details for predicting mortality of patients in the intensive care unit}\label{sec:additional_details_icu}

In this section, we describe our analysis of data on patients' stays in the intensive care unit (ICU) \citep{silva2012predicting} in more detail.

First, we computed the minimum, weighted mean, and maximum value of the 15 time-series variables presented in Table~\ref{table:icu_features}.
The weighted mean corresponds to a linear regression fit to the time series.
We then dropped any variable that had a proportion of missing values greater than or equal to 30\%.
This procedure resulted in a total of 37 features that we used to predict mortality: the summaries of the time-series variables along with all general descriptors measured at admission.

\begin{table*}
	\begin{center}
        \begin{tabular}{c|l|l|l}
        	Variable group & Variable name & Summary measure & Included in analysis${}^{1}$ \\
        	\toprule
        	Glasgow Coma Scale (GCS) & GCS & min, weighted mean${}^{2}$, max & Included\\
        	\hline
        	\multirow{6}{*}{Metabolic panel}  & HCO3 (serum bicarbonate) & min, weighted mean, max & Included\\
        	& BUN (blood urea nitrogen) & min, weighted mean, max & Included\\
        	& Na (serum sodium) & min, weighted mean, max & Included\\
        	& K (serum potassium) & min, weighted mean, max & Included\\
        	& Glucose & min, weighted mean, max & Included\\
        	\hline
        	Systolic arterial blood pressure (SysABP) & SysABP & min, weighted mean, max & Not included\\
        	\hline
        	\multirow{2}{*}{Complete blood count test} & White blood cell count (WBC) & min, weighted mean, max & Included\\
        	& Hematocrit (HCT) & min, weighted mean, max & Included\\
        	\hline
        	Temperature (Temp) & Temp & min, weighted mean, max & Included\\
        	\hline
        	Lactate & Lactate & min, weighted mean, max & Not included\\
        	\hline
        	Heart rate (HR) & HR & min, weighted mean, max & Included\\
        	\hline
        	\multirow{3}{*}{Respiration} & Respiration rate (RespRate) & min, weighted mean, max & Not included\\
        	& Mechanical ventilation (MechVent) & min, weighted mean, max & Not included\\
             & O2 (oxygen) & ratio of FiO2, PaO2 & Not included \\
        	\hline
        	Urine output & Urine & min, weighted mean, max & Included\\
        	\hline
        	\multirow{5}{*}{General descriptors} & Gender & identity${}^{3}$ & Included \\
        	& Height & identity & Not included \\
        	& Weight & identity & Included \\
        	& Age & identity & Included \\
        	& ICU admission type & identity & Included \\
        \end{tabular}
    \end{center}
    \caption{
    Available features in the MIMIC-II database, along with summary measures computed and an indicator of whether or not the feature was included in the analysis. Impossible values (e.g., $\leq 0$ for many variables) were dropped. Summary measures (minimum value, weighted mean, and maximum value) were computed for all time-series variables. Any variable with proportion missing $> 0.3$ was not included in the analysis, leading to a final analysis dataset with 37 variables. \\
    ${}^{1}$Features with a proportion of missing values $> 0.3$ were dropped from the analysis. \\
    ${}^{2}$Estimated response at mean measurement time from a linear regression of response on time. \\
    ${}^{3}$All general descriptors were measured a single time, at admission.
    }
    \label{table:icu_features}
\end{table*}

{\small
\bibliographystyle{plainnat}
\bibliography{brian-papers,shapley_bib}

\begin{thebibliography}{26}
\providecommand{\natexlab}[1]{#1}
\providecommand{\url}[1]{\texttt{#1}}
\expandafter\ifx\csname urlstyle\endcsname\relax
  \providecommand{\doi}[1]{doi: #1}\else
  \providecommand{\doi}{doi: \begingroup \urlstyle{rm}\Url}\fi

\bibitem[Arjovsky et~al.(2019)Arjovsky, Bottou, Gulrajani, and
  Lopez-Paz]{Arjovsky2019-cw}
M~Arjovsky, L~Bottou, I~Gulrajani, and D~Lopez-Paz.
\newblock Invariant risk minimization.
\newblock \emph{arXiv:1907.02893}, 2019.

\bibitem[Boyd and Vandenberghe(2004)]{boyd2004}
S~Boyd and L~Vandenberghe.
\newblock \emph{Convex optimization}.
\newblock Cambridge university press, 2004.

\bibitem[Breiman(2001)]{breiman2001}
L~Breiman.
\newblock Random forests.
\newblock \emph{Machine Learning}, 45\penalty0 (1):\penalty0 5--32, 2001.

\bibitem[Castro et~al.(2009)Castro, G{\'o}mez, and Tejada]{castro2009}
J~Castro, D~G{\'o}mez, and J~Tejada.
\newblock Polynomial calculation of the shapley value based on sampling.
\newblock \emph{Computers \& Operations Research}, 36\penalty0 (5):\penalty0
  1726--1730, 2009.

\bibitem[Charnes et~al.(1988)Charnes, Golany, Keane, and Rousseau]{charnes1988}
A~Charnes, B~Golany, M~Keane, and J~Rousseau.
\newblock Extremal principle solutions of games in characteristic function
  form: core, {C}hebychev and {S}hapley value generalizations.
\newblock In JK~Sengupta and GK~Kadekodi, editors, \emph{Econometrics of
  Planning and Efficiency}, pages 123--133. Springer, 1988.

\bibitem[Chen and Guestrin(2016)]{chen2016}
T~Chen and C~Guestrin.
\newblock {XGBoost: A Scalable Tree Boosting System}.
\newblock \emph{arXiv:1603.02754}, 2016.

\bibitem[Covert et~al.(2020)Covert, Lundberg, and Lee]{covert2020}
I~Covert, S~Lundberg, and SI~Lee.
\newblock Understanding global feature contributions through additive
  importance measures.
\newblock \emph{arXiv}, 2020.
\newblock \url{https://arxiv.org/abs/2004.00668}.

\bibitem[Dunning(2006)]{dunning2006}
AJ~Dunning.
\newblock A model for immunological correlates of protection.
\newblock \emph{Statistics in Medicine}, 25\penalty0 (9):\penalty0 1485--1497,
  2006.

\bibitem[Friedman(2001)]{friedman2001}
JH~Friedman.
\newblock Greedy function approximation: a gradient boosting machine.
\newblock \emph{Annals of Statistics}, 29\penalty0 (5):\penalty0 1189--1232,
  2001.

\bibitem[Garson(1991)]{garson1991}
DG~Garson.
\newblock Interpreting neural network connection weights.
\newblock \emph{Artificial Intelligence Expert}, 1991.

\bibitem[Gr{\"o}mping(2007)]{gromping2007}
U~Gr{\"o}mping.
\newblock Estimators of relative importance in linear regression based on
  variance decomposition.
\newblock \emph{The American Statistician}, 61\penalty0 (2):\penalty0 139--147,
  2007.

\bibitem[Hastie and Tibshirani(1990)]{hastie1990}
TJ~Hastie and RJ~Tibshirani.
\newblock \emph{Generalized {A}dditive {M}odels}, volume~43.
\newblock CRC Press, 1990.

\bibitem[Kingma and Ba(2014)]{kingma2014adam}
D~Kingma and J~Ba.
\newblock Adam: A method for stochastic optimization.
\newblock \emph{arXiv:1412.6980}, 2014.

\bibitem[Lundberg and Lee(2017)]{lundberg2017}
SM~Lundberg and S-I Lee.
\newblock A unified approach to interpreting model predictions.
\newblock In \emph{Advances in Neural Information Processing Systems}, 2017.

\bibitem[Lundberg et~al.(2020)Lundberg, Erion, Chen, DeGrave,
  et~al.]{lundberg2020}
SM~Lundberg, G~Erion, H~Chen, A~DeGrave, et~al.
\newblock From local explanations to global understanding with explainable {AI}
  for trees.
\newblock \emph{Nature Machine Intelligence}, 2\penalty0 (1):\penalty0
  2522--5839, 2020.

\bibitem[Murdoch et~al.(2019)Murdoch, Singh, Kumbier, Abbasi-Asl, and
  Yu]{murdoch2019}
WJ~Murdoch, C~Singh, K~Kumbier, R~Abbasi-Asl, and B~Yu.
\newblock Interpretable machine learning: definitions, methods, and
  applications.
\newblock \emph{arXiv:1901.04592}, 2019.

\bibitem[Nathans et~al.(2012)Nathans, Oswald, and Nimon]{nathans2012}
LL~Nathans, FL~Oswald, and K~Nimon.
\newblock Interpreting multiple linear regression: A guidebook of variable
  importance.
\newblock \emph{Practical Assessment, Research \& Evaluation}, 17\penalty0 (9),
  2012.

\bibitem[Owen and Prieur(2017)]{owen2017}
AB~Owen and C~Prieur.
\newblock On {S}hapley value for measuring importance of dependent units.
\newblock \emph{SIAM/ASA Journal on Uncertainty Quantification}, 5, 2017.
\newblock \doi{10.1137/16M1097717}.

\bibitem[Ribeiro et~al.(2016)Ribeiro, Singh, and Guestrin]{ribeiro2016}
MT~Ribeiro, S~Singh, and C~Guestrin.
\newblock Why should {I} trust you?: Explaining the predictions of any
  classifier.
\newblock In \emph{Proceedings of the 22nd ACM SIGKDD International Conference
  on Knowledge Discovery and Data Mining}, pages 1135--1144, 2016.

\bibitem[Shapley(1953)]{shapley1953}
LS~Shapley.
\newblock A value for n-person games.
\newblock \emph{Contributions to the Theory of Games}, 2\penalty0
  (28):\penalty0 307--317, 1953.

\bibitem[Silva et~al.(2012)Silva, Moody, Scott, Celi, and
  Mark]{silva2012predicting}
I~Silva, G~Moody, DJ~Scott, LA~Celi, and RG~Mark.
\newblock Predicting in-hospital mortality of icu patients: The
  physionet/computing in cardiology challenge 2012.
\newblock In \emph{Computing in Cardiology (CinC), 2012}. IEEE, 2012.

\bibitem[{\v{S}}trumbelj and Kononenko(2014)]{strumbelj2014}
E~{\v{S}}trumbelj and I~Kononenko.
\newblock Explaining prediction models and individual predictions with feature
  contributions.
\newblock \emph{Knowledge and Information Systems}, 41\penalty0 (3):\penalty0
  647--665, 2014.

\bibitem[van~der Vaart(2000)]{vandervaart2000}
AW~van~der Vaart.
\newblock \emph{Asymptotic {S}tatistics}, volume~3.
\newblock Cambridge University Press, 2000.

\bibitem[Williamson et~al.(2020{\natexlab{a}})Williamson, Gilbert, Carone, and
  Simon]{williamson2020a}
BD~Williamson, PB~Gilbert, M~Carone, and N~Simon.
\newblock Nonparametric variable importance assessment using machine learning
  techniques.
\newblock \emph{Biometrics}, (to appear), 2020{\natexlab{a}}.

\bibitem[Williamson et~al.(2020{\natexlab{b}})Williamson, Gilbert, Simon, and
  Carone]{williamson2020b}
BD~Williamson, PB~Gilbert, N~Simon, and M~Carone.
\newblock A unified approach for inference on algorithm-agnostic variable
  importance.
\newblock \emph{arXiv}, 2020{\natexlab{b}}.
\newblock \url{https://arxiv.org/abs/2004.03683}.

\bibitem[Zhang and Yu(2005)]{zhang2005}
T~Zhang and B~Yu.
\newblock Boosting with early stopping: Convergence and consistency.
\newblock \emph{The Annals of Statistics}, 33\penalty0 (4):\penalty0
  1538--1579, 2005.

\end{thebibliography}
}

\end{document}